\documentclass[aps,prl,reprint,groupedaddress,showpacs,amsmath]{revtex4-1}

\usepackage{amsthm}
\usepackage{hyperref}
\usepackage{graphicx}

\newtheorem{theorem}{Theorem}

\newcommand{\Ket}[1]{\left \lvert #1\right \rangle}

\newcommand{\BraKet}[2]{\left \langle #1 \middle \vert #2\right \rangle}

\newcommand{\KetBra}[2]{\left \lvert #1 \middle \rangle \middle
    \langle #2 \right \rvert}

\newcommand{\Proj}[1]{\KetBra{#1}{#1}}

\newcommand{\Abs}[1]{\left\lvert#1\right\rvert}

\newcommand{\QProb}[2]{\Abs{\BraKet{#1}{#2}}^2}

\begin{document}

\title{Maximally epistemic interpretations of the quantum state and
  contextuality}

\author{M. S. Leifer}
\email{matt@mattleifer.info}
\homepage{http://mattleifer.info}
\noaffiliation

\author{O. J. E. Maroney}
\email{owen.maroney@philosophy.ox.ac.uk}
\affiliation{Faculty of Philosophy, University of Oxford, 10 Merton
  Street, Oxford, OX1 4JJ, UK}
\altaffiliation[Mailing address]{Wolfson College, Linton Road, Oxford,
OX2 6UD, UK}

\date{\today}

\begin{abstract}
  We examine the relationship between quantum contextuality (in both
  the standard Kochen-Specker sense and in the generalised sense
  proposed by Spekkens) and models of quantum theory in which the
  quantum state is maximally epistemic.  We find that preparation
  noncontextual models must be maximally epistemic, and these in turn
  must be Kochen-Specker noncontextual.  This implies that the
  Kochen-Specker theorem is sufficient to establish both the
  impossibility of maximally epistemic models and the impossibility of
  preparation noncontextual models.  The implication from preparation
  noncontextual to maximally epistemic then also yields a proof of
  Bell's theorem from an EPR-like argument.
\end{abstract}


\maketitle

The nature of the quantum state has been debated since the early days
of quantum theory.  Is it a state of knowledge or information (an
\textit{epistemic} state), or is it a state of physical reality (an
\textit{ontic} state)?  One of the reasons for being interested in
this question is that many of the phenomena of quantum theory are
explained quite naturally in terms of the epistemic view of quantum
states \cite{Spekkens2007}.  For example, the fact that nonorthogonal
quantum states cannot be perfectly distinguished is puzzling if they
correspond to distinct states of reality.  However, on the epistemic
view, a quantum state is represented by a probability distribution
over ontic states, and nonorthogonal quantum states correspond to
overlapping probability distributions.  Indistinguishability is
explained by the fact that preparations of the two quantum states
would sometimes result in the same ontic state, and in those cases
there would be nothing existing in reality that could distinguish the
two.

Several theorems have recently been proved showing that the quantum
state must be an ontic state
\cite{Pusey2012,Colbeck2012,*Hall2011,*Hardy2012,*Schlosshauer2012}.
Most of these have been proved within the ontological models framework
\cite{Harrigan2010}, which generalizes the hidden variable approach
used to prove earlier no-go results, such as Bell's theorem
\cite{Bell1964} and the Bell-Kochen-Specker theorem \cite{Bell1966,
  Kochen1967}.  However, each of these new theorems rests on auxiliary
assumptions, of varying degrees of reasonableness.  For example, the
Pusey-Barrett-Rudolph theorem \cite{Pusey2012} assumes that the ontic
states of two systems prepared in a product state are statistically
independent, and the Colbeck-Renner result \cite{Colbeck2012} employs
a strong ``free choice'' assumption that rules out deterministic
theories a priori.  An explicit counterexample shows that these proofs
cannot be made to work without such auxiliary assumptions
\cite{Lewis2012}.

The requirement of ontic quantum states is perhaps the strongest
constraint on hidden variable theories that has been proved to date.
It immediately implies preparation contextuality (within the
generalised approach to contextuality of Spekkens
\cite{Spekkens2005}), a version of Bell's theorem, and that the ontic
state space must be infinite, with a number of parameters that
increases exponentially with Hilbert space dimension. See
\cite{Leifer2011} for a discussion of these implications.  However,
the auxiliary assumptions used in the proofs of the onticity of
quantum states carry over into these corollaries whereas the original
proofs of these results \cite{Spekkens2005, Bell1964, Hardy2004,
  *Montina2008} did not require them.  For this reason, it is
interesting to look for results addressing the ontic/epistemic
distinction that are weaker than completely ontic, but can be proved
without auxiliary assumptions, since such results may sit near the top
of a hierarchy of no-go theorems.

Recently, one of us introduced a stronger notion of what it means for
the quantum state to be epistemic and proved that it is incompatible
with the predictions of quantum theory without any auxiliary
assumptions \cite{Maroney2012}.  An ontological model is
\emph{maximally $\psi$-epistemic} if the quantum probability of
obtaining the outcome $\Ket{\phi}$ when measuring a system prepared in
the state $\Ket{\psi}$ is entirely accounted for by the overlap
between the corresponding probability distributions in the ontological
model.  This property is required if the epistemic explanation of the
indistinguishability of nonorthogonal states is to be strictly true.
It is satisfied by the $\psi$-epistemic model of two-dimensional
Hilbert spaces proposed by Kochen and Specker \cite{Kochen1967,
  Harrigan2010}, and its analog is satisfied by the epistemic toy
model of Spekkens \cite{Spekkens2007}.

In this letter, we explain how this stronger notion of quantum state
epistemicity relates to other no-go theorems, particularly the
traditional notion of noncontextuality used in proofs of the
Kochen-Specker theorem and Spekkens notion of preparation
contextuality.  Briefly, Kochen-Specker noncontextuality applies to
deterministic models, and says that if an outcome corresponding to
some projector is certain to occur in one measurement then outcomes
corresponding to the same projector in other measurements must also be
certain to occur.  Preparation noncontextuality says that preparation
procedures corresponding to the same density operator must be assigned
the same probability distribution.  Our results can
be summarized as:
\begin{multline}
  \text{Preparation noncontextual} \Rightarrow \text{Maximally }
  \psi\text{-epistemic} \\ \Rightarrow \text{Kochen-Specker noncontextual}.
\end{multline}
Both implications are strict, which we demonstrate with specific
examples of models that are Kochen-Specker noncontextual but not
maximally $\psi$-epistemic, and maximally $\psi$-epistemic but not
preparation noncontextual.  Since the no-go theorem for
maximally-epistemic models does not require auxiliary assumptions,
these implications provide a stronger proof of preparation
contextuality and Bell's theorem than those obtained from other no-go
theorems for $\psi$-epistemic models.

We are interested in ontological models that reproduce the quantum
predictions for a set of prepare-and-measure experiments.  The
experimenter can perform measurements of a set of orthonormal bases
$\mathcal{M} = \left \{ M_1,M_2, \ldots \right \}$ on the system.  Let
$\mathcal{P} = \cup_{M\in\mathcal{M}} M$ denote the set of quantum
states that occur in one or more of these bases.  Prior to the
measurement, the experimenter can prepare the system in any of the
states in $\mathcal{P}$.

An ontological model for $\mathcal{M}$ specifies a measure space
$(\Lambda, d\lambda)$ of ontic states.  Each state $\Ket{\psi} \in
\mathcal{P}$ is associated with a probability distribution
$\mu_{\psi}(\lambda)$ \footnote{More generally, $\Lambda$ is a
  measurable space and states are associated with probability measures
  $\nu_{\psi}$.  For ease of exposition, we have assumed that
  $\Lambda$ is equipped with a canonical measure $d \lambda$ with
  respect to which all the measures $\nu_{\psi}$ are absolutely
  continuous, so that we have well-defined densities $\mu_{\psi} =
  d\nu_{\psi}/d\lambda$.  This assumption is not strictly necessary.
  A more general measure theoretic treatment will appear in
  \cite{Leifer2013}.}  \footnote{More generally, a distribution is
  associated with the procedure for preparing a pure state rather than
  the state itself to allow for preparation contextuality.  Assuming
  preparation noncontextuality for pure states does no harm in the
  present context because our argument only establishes preparation
  contextuality for mixed states.} over $\Lambda$ and each measurement
$M \in \mathcal{M}$ is associated with a set of positive response
functions $\xi_M(\alpha|\lambda)$ that satisfy $\sum_{\Ket{\alpha} \in
  M} \xi_M(\alpha|\lambda) = 1$ for all $\lambda\in\Lambda$.  The
ontological model is required to reproduce the Born rule, which means
that $\forall \Ket{\psi}\in\mathcal{P}, M \in \mathcal{M},
\Ket{\alpha} \in M,$
\begin{equation}
  \label{eq:Reproduce}
  \int_{\Lambda} \xi_M(\alpha|\lambda) \mu_{\psi}(\lambda) d\lambda =
  \QProb{\alpha}{\psi}.
\end{equation}

For each state $\Ket{\psi} \in \mathcal{P}$, define $\Lambda_{\psi}
=\left \{ \lambda \middle | \mu_{\psi}(\lambda)>0 \right \}$.  We
assume that $\Lambda = \cup_{\Ket{\psi}\in\mathcal{P}}
\Lambda_{\psi}$, since otherwise there will be superfluous ontic
states that are never prepared.  Two important facts, of which we make
repeated use, are that: in order to reproduce $\QProb{\psi}{\psi}=1$
in Eq.~\eqref{eq:Reproduce}, for every $M$ that contains $\Ket{\psi}$
it must be the case that $\xi_M(\psi|\lambda) = 1$ almost everywhere
on $\Lambda_{\psi}$; and for all orthogonal $\Ket{\phi}\in M$, such
that $\QProb{\phi}{\psi}=0$, $\xi_M(\phi|\lambda) = 0$ almost
everywhere on $\Lambda_{\psi}$.  This implies that $\Lambda_{\psi}
\cap \Lambda_{\phi}$ is of measure zero for orthogonal $\Ket{\psi}$
and $\Ket{\phi}$.

A \emph{$\psi$-ontic} ontological model is one in which, for any pair
of \textit{nonorthogonal} quantum states $\Ket{\psi} \neq \Ket{\phi}$,
$\Lambda_{\psi}\cap\Lambda_{\phi}$ is of measure zero.  This means
that, if one knows the ontic state then the prepared quantum state can
be identified almost surely.  Conversely, if $\Lambda_{\psi} \cap
\Lambda_{\phi}$ has positive measure for some pair of states, then the
model is \emph{$\psi$-epistemic}.

An ontological model is \emph{maximally $\psi$-epistemic} if
$\int_{\Lambda_{\phi}} \mu_{\psi}(\lambda) d\lambda =
\QProb{\phi}{\psi}$ for every $\Ket{\psi}, \Ket{\phi} \in
\mathcal{P}$.  Since $\xi_M(\phi|\lambda) = 1$ almost everywhere on
$\Lambda_{\phi}$, then $\int_{\Lambda_{\phi}} \xi_M(\phi|\lambda)
\mu_{\psi}(\lambda) d\lambda = \QProb{\phi}{\psi}$. The probability of
obtaining the outcome $\Ket{\phi}$ when measuring a system prepared in
the state $\Ket{\psi}$ is entirely accounted for by the overlap
between $\mu_{\psi}$ and $\mu_{\phi}$.

The traditional notion of noncontextuality used in proofs of the
Kochen-Specker theorem is the combination of two conditions:

\begin{enumerate}
\item An ontological model is \emph{outcome deterministic} if
  $\xi_M(\alpha|\lambda) \in \{0,1\}$ almost everywhere on $\Lambda$,
  for all $M \in \mathcal{M}$, $\Ket{\alpha} \in M$.
\item An ontological model is \emph{measurement noncontextual} if,
  whenever $M,M'\in\mathcal{M}$ contain a common state $\Ket{\alpha}$,
  $\xi_M(\alpha|\lambda) = \xi_{M'}(\alpha|\lambda)$ almost everywhere
  on $\Lambda$.
\end{enumerate}

\begin{theorem}
  \label{thm:Measurement}
  If an ontological model of $\mathcal{M}$ is maximally $\psi$-epistemic
  then it is also outcome deterministic and measurement noncontextual.
\end{theorem}

\begin{proof}
  The proof closely parallels that of the ``quantum deficit theorem''
  \cite{Harrigan2007}.  As mentioned above, for any $\Ket{\phi} \in
  \mathcal{P}$, $\xi_M(\phi|\lambda) = 1$ almost everywhere on
  $\Lambda_{\phi}$ for every $M \in \mathcal{M}$ that contains
  $\Ket{\phi}$.  Hence, it is also equal to $1$ almost everywhere on
  $\Lambda_{\phi} \cap \Lambda_{\psi}$, since $\Lambda_{\psi}$ has
  positive measure.  In order to reproduce the Born rule, the
  ontological model must satisfy
  \begin{equation}
    \int_{\Lambda} \xi_M(\phi|\lambda) \mu(\lambda)_{\psi} d \lambda =
    \QProb{\phi}{\psi},
  \end{equation}
  but a maximally $\psi$-epistemic theory must also satisfy
  \begin{equation}
    \int_{\Lambda_{\phi}} \xi_M(\phi|\lambda) \mu_{\psi}(\lambda) d
    \lambda = \QProb{\phi}{\psi}.
  \end{equation}
  Given that these two equations must hold for all $\Ket{\psi} \in
  \mathcal{P}$, comparing them yields $\xi_M(\phi|\lambda) = 0$ almost
  everywhere on $\Lambda \backslash \Lambda_{\phi}$.  Thus, the model
  is outcome deterministic.  Since the same argument holds for every
  $M$ in which $\Ket{\phi}$ appears, the model is measurement
  noncontextual.
\end{proof}

The implication in this theorem is strict, i.e.\ there exist
Kochen-Specker noncontextual models that are not maximally
$\psi$-epistemic.  An example is provided by the Bell-Mermin model
\cite{Bell1966, Mermin1993}, in which $\mathcal{M}$ consists of all
orthonormal bases in a two-dimensional Hilbert space.  The ontic state
space of the model consists of the cartesian product of two copies of
the unit sphere $\Lambda = S_2 \times S_2$, and we denote the ontic
states as $\lambda = (\vec{\lambda_1},\vec{\lambda_2})$, where
$\vec{\lambda_j} \in S_2$.  For a state $\Ket{\psi}$, let $\vec{\psi}$
denote the corresponding Bloch vector.  The distribution associated
with $\Ket{\psi}$ in the ontological model is a product
$\mu_{\psi}(\lambda) = \mu_{\psi}(\vec{\lambda_1})
\mu_{\psi}(\vec{\lambda_2})$, where $\mu_{\psi}(\vec{\lambda_1}) =
\delta(\vec{\lambda_1} - \vec{\psi})$ is a point measure on
$\vec{\psi}$ \footnote{This does not satisfy the absolute continuity
  assumption, but a more technical version of
  theorem~\ref{thm:Measurement} holds without it \cite{Leifer2013}.}
and $\mu_{\psi}(\vec{\lambda_2}) = \frac{1}{4\pi}$ is the uniform
measure on $S_2$.  It is easy to see that this model is not maximally
$\psi$-epistemic because $\Lambda_{\psi} \cap \Lambda_{\phi} =
\emptyset$ for distinct $\Ket{\psi}$ and $\Ket{\phi}$ due to the
$\delta$-function term.  In fact the model is $\psi$-ontic.

The response functions of the model are
\begin{equation}
  \label{eq:BellResponse}
  \xi_M(\alpha|\lambda) = \Theta \left ( \vec{\alpha}\cdot \left (
      \vec{\lambda_1} + \vec{\lambda_2}
    \right ) \right ),
\end{equation}
where $\Theta$ is the Heaviside step function,
\begin{align}
  \Theta(x) & = 1, \qquad x > 0 \\
  & = 0, \qquad x \leq 0.
\end{align}
This model is outcome deterministic because $\Theta$ only takes values
$0$ and $1$, and it is measurement noncontextual because the right
hand side of Eq.~\eqref{eq:BellResponse} does not depend on $M$. It is
straightforward to check that the model reproduces the quantum
predictions.

In order to understand the connection between maximally
$\psi$-epistemic models and preparation contextuality, we need to
describe how (proper) mixtures are represented in ontological models.
Assume that, in addition to preparing the pure states in
$\mathcal{P}$, the experimenter can also prepare mixtures of them by
generating classical randomness (by flipping coins, rolling dice,
etc.) with probability distribution $p_j$ and then preparing a
different state $\Ket{\psi_j} \in \mathcal{P}$ depending on the
outcome, resulting in the density operator $\rho = \sum_j p_j
\Proj{\psi_j}$ \footnote{For present purposes, it would be sufficient
  to allow only 50/50 mixtures.}.  The classical randomness is assumed
to be independent of the ontic state of the quantum system, so that
the distribution over ontic states associated with preparing the
ensemble $\mathcal{E} = \{p_j, \Ket{\psi_j}\}$ is $\mu_{\mathcal{E}}
=\sum_j p_j \mu_{\psi_j}(\lambda)$.  An ontological model is
\emph{preparation noncontextual} if $\mu_{\mathcal{E}}(\lambda)$
depends only on the density operator $\rho$, and not on the specific
ensemble decomposition, $\mathcal{E} = \{p_j,\Ket{\psi_j}\}$, used to
prepare it.  Otherwise the model is \emph{preparation contextual}.

\begin{theorem}
  \label{thm:Preparation}
  Suppose an ontological model of $\mathcal{M}$ is not maximally
  $\psi$-epistemic, so that there exist states $\Ket{\psi}, \Ket{\phi}
  \in \mathcal{P}$ such that
  \begin{equation}
    \int_{\Lambda_{\phi}} \mu_{\psi}(\lambda) d \lambda <
    \QProb{\phi}{\psi}.
  \end{equation}
  Then, if $\mathcal{P}$ includes the states $\Ket{\psi^{\perp}},
  \Ket{\phi^{\perp}}$ that satisfy $\BraKet{\psi^{\perp}}{\psi} = 0$
  and $\BraKet{\phi^{\perp}}{\phi}=0$, and are in the subspace spanned
  by $\Ket{\psi}$ and $\Ket{\phi}$, then the model is also preparation
  contextual.
  
\end{theorem}

\begin{proof}
  By Eq.~\eqref{eq:Reproduce}, for any $M$ containing $\Ket{\phi}$,
  \begin{align}
    \QProb{\phi}{\psi} & = \int_{\Lambda} \xi_M(\phi|\lambda)
    \mu_{\psi}(\lambda) d \lambda \\
    & \geq \int_{\Lambda_{\phi}} \xi_M(\phi|\lambda)\mu_{\psi}(\lambda)
    d \lambda
    & = \int_{\Lambda_{\phi}} \mu_{\psi}(\lambda) d\lambda,
  \end{align}
  where the last line follows because $\xi_M(\phi|\lambda) = 1$ almost
  everywhere on $\Lambda_{\phi}$.  By assumption, the inequality must
  be strict, so we have
  \begin{equation}
    \int_{\Lambda_{\phi}} \xi_M(\phi|\lambda)\mu_{\psi}(\lambda)
    d \lambda < \int_{\Lambda} \xi_M(\phi|\lambda)
    \mu_{\psi}(\lambda) d \lambda
  \end{equation}

  This means that there must be a set $\Omega$ of ontic states that is
  disjoint from $\Lambda_{\phi}$, is assigned nonzero probability by
  $\mu_{\psi}$, and is such that $\xi_M(\phi|\lambda) > 0$ for
  $\lambda \in \Omega$.  Now, consider the two mixed preparations:
  \begin{enumerate}
  \item Prepare $\Ket{\psi}$ with probability $1/2$ and
    $\Ket{\psi^{\perp}}$ with probability $1/2$.
  \item Prepare $\Ket{\phi}$ with probability $1/2$ and
    $\Ket{\phi^{\perp}}$ with probability $1/2$.
  \end{enumerate}
  The resulting density operators, $\rho_1$ and $\rho_2$, satisfy
  $\rho_1 = \rho_2 = \frac{1}{2} \Pi$, where $\Pi$ is the projector
  onto the subspace spanned by $\Ket{\psi}$ and $\Ket{\phi}$.  Let
  $\Lambda_1 = \Lambda_{\psi} \cup \Lambda_{\psi^{\perp}}$ and
  $\Lambda_2 = \Lambda_{\phi} \cup \Lambda_{\phi^{\perp}}$ be the
  supports of the corresponding distributions, $\mu_1 =
  \frac{1}{2}(\mu_{\psi} + \mu_{\psi^{\perp}})$ and $\mu_2 =
  \frac{1}{2}(\mu_{\phi} + \mu_{\phi^{\perp}})$, in the ontological
  model.  Now, $\Lambda_1 \cap \Omega$ is assigned nonzero probability
  by $\mu_1$, whereas $\mu_2$ assigns probability zero to $\Omega$.
  This is because $\Lambda_{\phi}$ is disjoint from $\Omega$ by
  definition and $\mu_{\phi^{\perp}}$ must assign zero probability any
  set of ontic states that assign nonzero probability to $\Ket{\phi}$
  in a measurement of any orthonormal basis that contains it.  Hence
  $\mu_1$ and $\mu_2$ must be distinct because their supports differ
  by a set of positive measure.
\end{proof}

A simple corollary of this theorem is that, whenever the states
$\Ket{\psi^{\perp}}$ and $\Ket{\phi^{\perp}}$ are in $\mathcal{P}$ for
every $\Ket{\psi}, \Ket{\phi} \in \mathcal{P}$, then any preparation
noncontextual ontological model is also maximally $\psi$-epistemic.
As in the case of Kochen-Specker contextuality, this implication is
strict, i.e.\ there are maximally $\psi$-epistemic models that are
preparation contextual.  An example is provided by the Kochen-Specker
model \cite{Kochen1967}, which again takes $\mathcal{M}$ to be all
orthonormal bases in a two-dimensional Hilbert space.  This time, the
ontic state space is just a single copy of the unit sphere $\Lambda =
S_2$ and the ontic states are unit vectors $\vec{\lambda} \in S_2$.
The probability distribution associated with a quantum state
$\Ket{\psi}$ is
\begin{equation}
  \mu_{\psi}(\vec{\lambda}) = \frac{1}{\pi} \Theta(\vec{\psi} \cdot
  \vec{\lambda}) \vec{\psi} \cdot \vec{\lambda}
\end{equation}
and the response function associated with a quantum
state $\Ket{\phi}$ is
\begin{equation}
  \xi_M(\phi|\lambda) = \Theta (\vec{\phi} \cdot \vec{\lambda}).
\end{equation}
It is straightforward to check that this model reproduces the quantum
predictions.

The model is maximally $\psi$-epistemic because $\Lambda_{\phi} =
\{\lambda | \Theta(\vec{\phi}\cdot\vec{\lambda}) = 1\}$ and thus
\begin{align}
  \QProb{\phi}{\psi} & = \int_{\Lambda}
  \xi_M(\phi|\lambda)\mu_{\psi}(\vec{\lambda}) d \lambda\\
  & = \int_{\Lambda} \Theta(\vec{\phi} \cdot
  \vec{\lambda}) \mu_{\psi}(\vec{\lambda}) d \lambda\\
  & = \int_{\Lambda_{\phi}} \mu_{\psi}(\vec{\lambda}) d \lambda.
\end{align}

On the other hand, the model is preparation contextual as can be seen
by considering the two preparations:
\begin{enumerate}
  \item Prepare $\Ket{+z}$ with probability $1/2$ and $\Ket{-z}$ with
    probability $1/2$.
  \item Prepare $\Ket{+x}$ with probability $1/2$ and $\Ket{-x}$ with
    probability $1/2$.
\end{enumerate}
Both preparations correspond to the maximally mixed state, but the
distributions $\frac{1}{2}(\mu_{+z} + \mu_{-z})$ and
$\frac{1}{2}(\mu_{+x} + \mu_{-x})$ are different.  In particular, both
$\mu_{+z}$ and $\mu_{-z}$ are zero on the equator whereas $\mu_{+x}$
and $\mu_{-x}$ are both nonzero here.

Theorem~\ref{thm:Measurement} implies that any proof of the
Kochen-Specker theorem is sufficient to establish that maximally
$\psi$-epistemic models are impossible for Hilbert spaces of dimension
greater than two.  Unlike the proof in \cite{Maroney2012}, however,
this does not establish a bound on how close to maximally
$\psi$-epistemic one can get.  Further, the Kochen-Specker theorem
allows a finite precision loophole \cite{Meyer1999, *Kent1999,
  *Clifton2001, *Barrett2004, *Hermens2011} that can be exploited to
allow noncontextual theories to get arbitrarily close to quantum
statistics, so it seems unlikely that this proof could be made robust
against experimental error.

Combining the two theorems also shows that the Kochen-Specker theorem
is enough to establish preparation contextuality.  Whilst it was known
that preparation noncontextuality implies outcome determinism for
models of quantum theory \cite{Spekkens2005}, it is a novel
implication that it also implies measurement noncontextuality.  This
demonstrates that the ontic/epistemic distinction is useful for
understanding the relationship between existing no-go theorems.

Finally, the type of preparation contextuality established by our
results can be used to prove Bell's theorem.  Briefly, if $\Ket{\psi}$
and $\Ket{\phi}$ are states such that
\begin{equation}
  \int_{\Lambda_{\phi}} \mu_{\psi}(\lambda) d \lambda <
  \QProb{\phi}{\psi}
\end{equation}
then we can demonstrate nonlocality using a maximally
entangled state $\frac{1}{\sqrt{2}} \left ( \Ket{\psi}_A\Ket{\psi}_B +
  \Ket{\psi^{\perp}}_A\Ket{\psi^{\perp}}_B \right )$.  Since the
reduced density matrix on Bob's system is
\begin{align}
  \rho & = \frac{1}{2} \left ( \Proj{\psi} +
  \Proj{\psi^{\perp}} \right ) \\
  & = \frac{1}{2} \left ( \Proj{\phi} +
  \Proj{\phi^{\perp}}\right ),
\end{align}
by the Schr{\"o}dinger-HJW theorem \cite{Schrodinger1936,
  *Hughston1993} there are two measurements that Alice can perform,
the first of which will collapse Bob's system to $\Ket{\psi}$ or
$\Ket{\psi^{\perp}}$ with 50/50 probabilities, and the second of which
will collapse it to $\Ket{\phi}$ or $\Ket{\phi^{\perp}}$ with 50/50
probabilities.  However, theorem~\ref{thm:Preparation} establishes
that these two ensembles cannot correspond to the same probability
distribution over ontic states.  Thus, the distribution on Bob's side
must depend on the choice of measurement that Alice makes, which
implies Bell nonlocality.  This argument generalizes the proof of
\cite{Harrigan2010}, which showed that local theories would have to be
$\psi$-epistemic.  In fact, we see that they would have to be
\emph{maximally} $\psi$-epistemic.  Filling in the formal details of
this argument can be done in a similar way to \cite{Harrigan2010}.

Whilst the impossibility of a maximally $\psi$-epistemic theory
clarifies what can be proved about contextuality based on the
ontic/epistemic distinction without auxiliary assumptions, it is not
sufficient to establish the constraints on the size of the ontic state
space that follow from having fully ontic quantum states
\cite{Hardy2004, *Montina2008}.  If one could prove, without auxiliary
assumptions, that the support of every distribution in an ontological
model must contain a set of states that are not shared by the
distribution corresponding to any other quantum state, then these
results would follow.  Whether this can be proved is an important open
question.

\begin{acknowledgments}
  We would like to thank Chris Timpson for helpful discussions.  OJEM
  is supported by the John Templeton Foundation.
\end{acknowledgments}

\bibliography{ContextualityPaper}

\end{document}